\documentclass[11pt]{article}
\usepackage[T1]{fontenc}
\usepackage{lmodern}
\usepackage{amsmath,amssymb,amsthm}
\usepackage{fullpage}
\usepackage{microtype}
\usepackage{aliascnt}
\usepackage{booktabs}
\usepackage{xcolor}
\usepackage{tcolorbox}
\definecolor{linkblue}{RGB}{28,67,112}
\usepackage[colorlinks=true,allcolors=linkblue]{hyperref}
\usepackage[nameinlink,capitalize,noabbrev]{cleveref}

\newaliascnt{lemma}{theorem}
\newtheorem{lemma}[lemma]{Lemma}
\aliascntresetthe{lemma}
\crefname{lemma}{lemma}{lemmas}
\Crefname{lemma}{Lemma}{Lemmas}
\newaliascnt{corollary}{theorem}

\aliascntresetthe{corollary}
\crefname{corollary}{corollary}{corollaries}
\Crefname{corollary}{Corollary}{Corollaries}
\newtheorem*{claim}{Claim}
\theoremstyle{definition}
\newaliascnt{definition}{theorem}

\aliascntresetthe{definition}
\crefname{definition}{definition}{definitions}
\Crefname{definition}{Definition}{Definitions}
\usepackage{thm-restate}
\numberwithin{equation}{section}
\newcommand{\Ostar}[1]{O^*(#1)}
\newcommand{\sca}{\textnormal{\textsc{SCA}}}
\newcounter{problem}[section]
\renewcommand{\theproblem}{\thesection.\arabic{problem}}

\crefname{problem}{problem}{problems}
\Crefname{problem}{Problem}{Problems}
\newenvironment{boxproblem}[2]{%
  \begin{center}
    \begin{tcolorbox}[
        width=0.96\linewidth,
        colback=white,
        colframe=black!35,
        boxrule=0.45pt,
        arc=4pt,
        outer arc=4pt,
        boxsep=0pt,
        left=7pt,
        right=7pt,
        top=6pt,
        bottom=6pt]
      \refstepcounter{problem}\label{#2}%
      {\centering\normalfont\normalsize\scshape Problem~\theproblem: #1\par}
      \vspace{3pt}
      {\color{black!35}\hrule height 0.45pt}
      \vspace{5pt}
      \normalfont
}{%
    \end{tcolorbox}
  \end{center}
}

\title{Single-Exponential Algorithms and a Polynomial Kernel\\
  for Strong Connectivity Augmentation}
\author{Tomohiro Koana\thanks{Graduate School of Information Science and Technology, The University of Tokyo, Japan.\newline Email: \texttt{tomohiro.koana@gmail.com}.}
\and Soh Kumabe\thanks{CyberAgent, Inc., Tokyo, Japan. Email: \texttt{kumabe\_soh@cyberagent.co.jp}.}}
\date{}
\hypersetup{pdftitle={Single-Exponential Algorithms and a Polynomial Kernel for Strong Connectivity Augmentation},pdfauthor={Tomohiro Koana and Soh Kumabe}}

\begin{document}
\maketitle
\begin{abstract}
\textsc{Strong Connectivity Augmentation} (\sca) asks whether a directed acyclic graph can be made strongly connected by adding at most $k$ prescribed links whose total weight is within a given budget.
Klinkby, Misra, and Saurabh (SODA 2021) gave an $\Ostar{2^{O(k\log k)}}$-time algorithm and asked whether the problem admits a single-exponential parameterized algorithm and a polynomial kernel.
We answer both questions affirmatively: \sca\ can be solved in $\Ostar{9^k}$ time and admits a polynomial kernel with $O(k^4)$ vertices and $O(k^{16})$ bits.
For unweighted \sca, we obtain $\Ostar{4^k}$ time and a kernel with $O(k^3)$ vertices.
Our algorithms are based on a particularly simple reduction to
 \textsc{Strongly Connected Spanning Subgraph} with two edge costs.
\end{abstract}

\section{Introduction}
\label{sec:introduction}

\textsc{Strong Connectivity Augmentation} (\sca) asks how to add directed edges to a digraph $G$ so that every vertex can reach every other vertex.
We write $uv$ for an edge directed from $u$ to $v$.
We call the edges in $G$ \emph{base edges} and the allowed additional edges \emph{links}.
Formally, we study the following problem on a directed acyclic graph (DAG).\footnote{A general digraph can be reduced to a DAG by contracting each strongly connected component.}

\begin{boxproblem}{\textsc{Strong Connectivity Augmentation} (\sca)}{prob:sca}
\textbf{Input:} A DAG $G$, a link set $L\subseteq V(G)\times V(G)$, weights $w\colon L\to\mathbb Z_{\ge0}$, and integers $k,t\ge0$.

\smallskip
\textbf{Task:} Decide whether there exists $F\subseteq L$ with $|F|\le k$ and $w(F):=\sum_{e\in F}w(e)\le t$ such that $G+F:=(V(G),E(G)\cup F)$ is strongly connected.

\smallskip
\textbf{Parameter:} $k$.
\end{boxproblem}

Even an edge-free input captures \textsc{Directed Hamiltonian Cycle}: for $n\ge3$, a strong digraph on $n$ vertices with at most $n$ edges must be a directed Hamiltonian cycle.
Thus the number $k$ of selected links provides a parameter for a problem that is already difficult on very simple base digraphs.

Connectivity augmentation has recently received considerable attention in parameterized complexity.
For undirected graphs, Carmesin and Ramanujan~\cite{carmesin2026} gave an $\Ostar{k^{O(k)}}$-time algorithm for prescribed-link augmentation to target vertex-connectivity $1\le\lambda\le4$.\footnote{The notation $\Ostar{\cdot}$ suppresses factors polynomial in the input bit length, with degree independent of the parameters in the argument.}
Here $k$ is the number of added links.
Korhonen and Thorup~\cite{korhonen2026} obtained $\Ostar{(k+\lambda)^{O(k)}}$ time for arbitrary target vertex-connectivity $\lambda$ and $\Ostar{k^{O(k)}}$ time for arbitrary target edge-connectivity.
For target vertex-connectivity two, Koana and Kumabe~\cite{koana2026} gave an $\Ostar{36^k}$-time algorithm for the unweighted case.
Their algorithm also handles positive integer link costs with pseudo-polynomial dependence on the maximum cost.

For \sca, Guo and Uhlmann~\cite{guo2010} asked in 2010 whether the problem is fixed-parameter tractable in $k$.
The question remained open for more than a decade, until Klinkby, Misra, and Saurabh~\cite{klinkby2021} gave an $\Ostar{2^{O(k\log k)}}$-time algorithm in 2021.
They asked whether single-exponential parameterized algorithms and polynomial kernels exist, for both the unweighted and weighted cases.

\subsection{Our contributions}
\label{subsec:results}

We resolve both open questions of Klinkby et al.~\cite{klinkby2021} affirmatively.
Our algorithms are based on a simple reduction to \textsc{Strongly Connected Spanning Subgraph} (\textsc{SCSS}) with two costs.
For a directed multigraph $H$, write $V(H)$ and $E(H)$ for its vertex and edge sets.

\begin{boxproblem}{\textsc{Strongly Connected Spanning Subgraph} (\textsc{SCSS}) with two costs}{prob:pair-scss}
\textbf{Input:} A directed multigraph $J$, costs $\ell,w\colon E(J)\to\mathbb Z_{\ge0}$, and integers $k,t\ge0$.

\smallskip
\textbf{Task:} Decide whether there exists $B\subseteq E(J)$ with $\ell(B):=\sum_{e\in B}\ell(e)\le k$ and $w(B):=\sum_{e\in B}w(e)\le t$ such that $(V(J),B)$ is strongly connected.
\end{boxproblem}

\begin{restatable}{theorem}{reductiontheorem}
\label{thm:reduction}
An instance of \sca\ can be reduced in polynomial time to an equivalent instance of \textnormal{\textsc{SCSS}} with two costs on at most $2k$ vertices.
\end{restatable}

\paragraph{Reduction.}
We sketch the reduction for an instance $\mathcal I=(G,L,w,k,t)$.

Let $S$ and $T$ be the sources and sinks of $G$, with isolated vertices in both sets.
The vertices of $Z:=S\cup T$ are \emph{terminals}; all others are \emph{nonterminals}.
Every source needs an entering link and every sink needs a leaving link, so we reject if $|S|>k$ or $|T|>k$.
Henceforth, assume $|S|,|T|\le k$, which gives $|Z|\le2k$.
For distinct $u,v\in Z$ and $0\le h\le k$, let $d_h(u,v)$ be the minimum total link weight of a directed $u$-to-$v$ path in $G+L$ using at most $h$ links, with value infinity if no such path exists.
These values can be computed in polynomial time by \Cref{lem:bounded-link-distances}.
Construct a directed multigraph $J$ with vertex set $Z$ and an edge $uv$ with costs $(\ell(uv),w(uv))=(h,d_h(u,v))$ for every such $u,v,h$ with $d_h(u,v)<\infty$.
Let $\mathcal J=(J,\ell,w,k,t)$ be the \textsc{SCSS} instance with two costs on $J$.

We briefly sketch why the reduction is correct.
One direction is straightforward: every vertex of $G$ is reachable from a source and can reach a sink, so making all terminals mutually reachable makes the augmented graph strongly connected.
If $\mathcal J$ is a yes-instance, replacing each selected edge $uv$ of first cost $h$ by a path attaining $d_h(u,v)$ selects at most $k$ links of total weight at most $t$ and makes the terminals mutually reachable, so $\mathcal I$ is a yes-instance.
For the other direction, Mader's directed splitting theorem gives a strongly connected multigraph on $Z$ represented by paths that use each selected link at most once in total, thereby respecting both budgets.
See \Cref{sec:terminal-reduction} for details.

\paragraph{Algorithms and kernels.}
For an unweighted instance $\mathcal I=(G,L,k)$, let $H:=G+L$, with base edges of cost zero and links of cost one, and write $d_H(u,v)$ for the shortest-path distance from $u$ to $v$ in $H$.
The reduction simplifies to a single-cost instance $\mathcal J=(J,d_H,k)$, where $J$ has an edge $uv$ of cost $d_H(u,v)$ for every distinct $u,v\in Z$ such that $v$ is reachable from $u$ in $H$.
We solve unweighted \sca\ by reducing it to \textsc{SCSS} on at most $2k$ vertices and applying the dynamic program of Jabal Ameli et al.~\cite{ameli2026}.

\begin{restatable}{theorem}{unweightedtheorem}
\label{thm:unweighted}
Unweighted \sca\ can be solved in $\Ostar{4^k}$ time.
\end{restatable}

For \sca, we adapt the dynamic program to the two edge costs, obtaining the following running time.

\begin{restatable}{theorem}{weightedtheorem}
\label{thm:weighted}
\sca\ can be solved in $\Ostar{9^k}$ time.
\end{restatable}

The reduction also yields polynomial kernels.
Recall that a \emph{kernelization} is a polynomial-time transformation into an equivalent instance of the same parameterized problem whose size and parameter are bounded by functions of the original parameter.
The unweighted kernel retains links on shortest paths between terminals and compresses base-edge reachability to the terminals and retained link endpoints.

\begin{restatable}{theorem}{unweightedkernel}
\label{thm:kernel-unweighted}
Unweighted \sca\ admits a polynomial kernel with $O(k^3)$ vertices.
\end{restatable}

For \sca, we retain a minimum-weight path for each terminal pair and each bound $0\le h\le k$ on the number of links, whenever such a path exists.
We then apply Frank--Tardos weight reduction~\cite{frank1987} to compress the retained link weights and the budget.

\begin{restatable}{theorem}{weightedkernel}
\label{thm:kernel}
\sca\ admits a polynomial kernel with $O(k^4)$ vertices and $O(k^{16})$ bits.
\end{restatable}

\subsection{Related work}
\label{subsec:related-work}

For unweighted \sca\ with every missing edge available as a link, a minimum augmentation can be found in linear time using the algorithm of Eswaran and Tarjan~\cite{eswaran1976} with the corrected construction of Raghavan~\cite{raghavan2005}.
Frederickson and Ja'Ja'~\cite{frederickson1981} gave constant-factor approximation algorithms for weighted connectivity augmentation problems, including \sca.
Bessy et al.~\cite{bessy2026} gave a single-exponential parameterized algorithm for \textsc{Plane Strong Connectivity Augmentation}, where the input is a connected plane oriented graph and added edges must preserve the embedding and cannot create a pair of oppositely directed edges.

Earlier work on undirected augmentation includes the fixed-parameter algorithms and polynomial kernels of Marx and V\'egh~\cite{marx2015} for increasing edge-connectivity by one, parameterized by the number of added links, also allowing link costs.
Basavaraju et al.~\cite{basavaraju2014} improved the running time to $\Ostar{9^k}$ for this weighted problem by reducing it to \textsc{Steiner Tree} with vertex weights.
Nutov~\cite{nutov2024} obtained $\Ostar{9^k}$-time algorithms for increasing vertex-connectivity from two to three and for increasing rooted vertex-connectivity from $\lambda-1$ to $\lambda$ for arbitrary $\lambda$.

Jabal Ameli et al.~\cite{ameli2026} gave an $\Ostar{2^n}$-time algorithm for unweighted \textsc{SCSS} on $n$ vertices.
Our algorithms combine the reduction to at most $2k$ terminals with their \textsc{SCSS} algorithm, adapted to edge costs and a bound on the number of links (\Cref{sec:algorithm}).

For strongly connected inputs, the \textsc{Minimum Equivalent Digraph} problem, which asks for a smallest spanning subgraph preserving all reachability relations, is precisely unweighted \textsc{SCSS}.
Khuller, Raghavachari, and Young~\cite{khuller1995} gave a $(\pi^2/6+\varepsilon)$-approximation for every fixed $\varepsilon>0$.
Vetta~\cite{vetta2001} obtained a $3/2$-approximation.
Bang-Jensen and Yeo~\cite{bangjensen2008} studied a different parameterization for \textsc{SCSS}: deciding whether an $n$-vertex strong digraph has a strong spanning subgraph with at most $2n-2-p$ edges is fixed-parameter tractable in $p$.

\paragraph{Organization.}
The rest of this paper is organized as follows.
\Cref{sec:terminal-reduction} proves \Cref{thm:reduction} by giving the full reduction to \textsc{SCSS} with two costs and establishing its correctness.
\Cref{sec:applications} gives the unweighted and weighted algorithms in \Cref{sec:algorithm} (\Cref{thm:unweighted,thm:weighted}) and the polynomial kernels in \Cref{sec:kernel} (\Cref{thm:kernel-unweighted,thm:kernel}).
\Cref{sec:conclusion} discusses possible approximation applications.

\section{Reduction to \textnormal{\textsc{SCSS}} with two costs}
\label{sec:terminal-reduction}

Let $\mathcal I=(G,L,w,k,t)$ and $\mathcal J=(J,\ell,w,k,t)$ be the augmentation and \textsc{SCSS} instances in \Cref{subsec:results}.

First, we verify that the distances can be computed in polynomial time.

\begin{lemma}
\label{lem:bounded-link-distances}
For distinct $u,v\in Z$ and $0\le h\le k$, the value $d_h(u,v)$ can be computed in polynomial time.
\end{lemma}

\begin{proof}
Construct a DAG with layers $0,\ldots,h$, each containing a copy of $G$ with base edges of weight zero.
For each link $xy\in L$ and $0\le i<h$, add an edge of weight $w(xy)$ from $x$ in layer $i$ to $y$ in layer $i+1$.
Then $d_h(u,v)$ is the minimum distance from $u$ in layer $0$ to $v$ in any layer $0,\ldots,h$.
Indeed, each path in $G+L$ using at most $h$ links lifts to the layered DAG, and a path in the layered DAG projects to a walk from which closed subwalks can be deleted without increasing its link count or weight.
The layered DAG has polynomial size, so its shortest-path distances can be computed in polynomial time.
\end{proof}

For every distinct $u,v\in Z$ and $0\le h\le k$ with $d_h(u,v)<\infty$, let $P_{uv,h}$ be a $u$-to-$v$ path in $G+L$ using at most $h$ links and attaining $d_h(u,v)$.

\paragraph*{Correctness.}
We first prove the direction from $\mathcal J$ to $\mathcal I$ by expanding the selected edges into their chosen paths.
Call $F\subseteq L$ \emph{feasible} if $G+F$ is strongly connected.

\begin{samepage}
\begin{lemma}
\label{lem:scss-to-augmentation}
If $\mathcal J$ is a yes-instance, then $\mathcal I$ is a yes-instance.
\end{lemma}

\begin{proof}
Suppose $\mathcal J$ is a yes-instance, and let $(Z,B)$ be an \textsc{SCSS} of $J$ with $\ell(B)\le k$ and $w(B)\le t$.
For every edge $e=uv\in B$ with $\ell(e)=h$, take the path $P_{uv,h}$, and let $F_B$ be the union of the link sets of these paths.
Each path uses at most $\ell(e)=h$ links and has weight $w(e)=d_h(u,v)$.
Since all weights are nonnegative, $|F_B|\le \ell(B)$ and $w(F_B)\le w(B)$, even if chosen paths share links.
Expanding directed paths in $(Z,B)$ into the chosen paths shows that every terminal can reach every other terminal in $G+F_B$.
Now let $x,y\in V(G)$.
Since $G$ is acyclic, it contains paths from $x$ to a sink $v$ and from a source $u$ to $y$.
Concatenating the $x$-to-$v$ path, a $v$-to-$u$ walk in $G+F_B$, and the $u$-to-$y$ path gives an $x$-to-$y$ walk, so $F_B$ is feasible.

Thus $F_B$ satisfies both bounds of $\mathcal I$, so $\mathcal I$ is a yes-instance.
\end{proof}
\end{samepage}

For the forward direction, we use Mader's directed splitting theorem to represent feasible augmentations by terminal paths.
For a directed multigraph $H$, let $d_H^-(v)$ and $d_H^+(v)$ be the indegree and outdegree of $v$, counting multiplicities; a loop contributes one to each.
A vertex $v$ is \emph{balanced} if $d_H^-(v)=d_H^+(v)$.
\emph{Splitting} a pair $xv$, $vy$ deletes those two edges and adds a new edge $xy$; parallel edges and loops are permitted.

\medskip
\noindent\begin{minipage}{\linewidth}
\begin{lemma}[{Mader's directed splitting theorem; \cite[Corollary~7.5.3]{bangjensen2007}}]
\label{lem:splitting}
Let $H$ be a loopless strongly connected directed multigraph, and let $v\in V(H)$ be balanced with $|V(H)\setminus\{v\}|\ge2$.
There is a bijection from the edges entering $v$ to the edges leaving $v$ such that splitting every matched pair and deleting $v$ yields a strongly connected directed multigraph on $V(H)\setminus\{v\}$.
\end{lemma}
\end{minipage}

\begin{lemma}
\label{lem:terminal-reduction}
If $\mathcal I$ is a yes-instance, then $\mathcal J$ is a yes-instance.
\end{lemma}

\begin{proof}
Suppose $\mathcal I$ is a yes-instance, and let $F\subseteq L$ be a feasible augmentation with $|F|\le k$ and $w(F)\le t$.

\begin{claim}
There exist a loopless strongly connected directed multigraph $M$ on $Z$ and, for every edge $e=uv\in E(M)$, a directed $u$-to-$v$ path $P_e$ in $G+F$ such that each link in $F$ occurs at most once across these paths.
\end{claim}

\begin{proof}[Proof of the claim]
For each edge $e=uv$ of $G+F$, let $W_e$ be the walk obtained by concatenating a path in $G$ from a source to $u$, the edge $e$, and a path in $G$ from $v$ to a sink.
Define a directed multigraph $H$ with vertex set $V(G)$ and edge multiset consisting of one distinct copy for every edge occurrence in every walk $W_e$, for $e\in E(G+F)$.
Since each edge $e$ of $G+F$ occurs in $W_e$, $H$ contains $G+F$ as a spanning subgraph and is therefore strongly connected.
The two paths concatenated with $e$ to form $W_e$ lie in $G$ and hence contain no links of $F$.
Thus each link $e\in F$ occurs only in $W_e$, exactly once, so $H$ contains exactly one copy of $e$.
Every nonterminal is balanced: the walks have endpoints in $Z$, and each internal visit contributes one entering and one leaving edge.

Let $H'$ be a copy of $H$.
While $H'$ contains a nonterminal, choose $v\in V(H')\setminus Z$, split the pairs given by \Cref{lem:splitting}, and delete $v$ and any resulting loops.
Splitting preserves the degrees of every remaining vertex, and deleting a loop decreases both degrees by one, so every remaining nonterminal stays balanced.
Since $|Z|\ge2$, \Cref{lem:splitting} applies in every iteration.
When $V(H')=Z$, let $M:=H'$, which is loopless and strongly connected.

Reversing the splittings, we obtain a directed $u$-to-$v$ walk in $H$ for every edge $e=uv\in E(M)$.
These walks use each edge of $H$ at most once in total: initially each edge represents itself, each split concatenates two representing walks, and loop deletion discards a representing walk.
Replace the copies by their original edges and delete closed subwalks to obtain the paths $P_e$ in $G+F$ with the same endpoints.
Each link has only one copy in $H$, so each link occurs at most once across all these paths, as required.
\end{proof}

To construct an \textsc{SCSS}, for every edge $e=uv\in E(M)$, let $h$ be the number of links on $P_e$ and select the edge $uv$ of $J$ with costs $(h,d_h(u,v))$, whose second cost is at most the weight of $P_e$.
Let $B\subseteq E(J)$ be the set of selected edges.
Since $M$ is strongly connected, $(Z,B)$ is an \textsc{SCSS} of $J$.
By the claim, summing the link counts and weights over the paths $P_e$ counts each link of $F$ at most once, so $\ell(B)\le |F|\le k$ and $w(B)\le w(F)\le t$.
\end{proof}

\Cref{lem:bounded-link-distances,lem:scss-to-augmentation,lem:terminal-reduction} together prove \Cref{thm:reduction}.

\section{Single-exponential algorithms and a polynomial kernel}
\label{sec:applications}

In this section, we prove \Cref{thm:unweighted,thm:weighted,thm:kernel-unweighted,thm:kernel}.

\subsection{Single-exponential algorithms}
\label{sec:algorithm}

Given an instance $\mathcal I=(G,L,w,k,t)$ of \sca, apply \Cref{thm:reduction} to obtain an equivalent instance $\mathcal J=(J,\ell,w,k,t)$ of \textsc{SCSS} with two costs.
Write $Z=S\cup T=V(J)$ and $n:=|Z|\le2k$.

We recall the dynamic programming of Jabal Ameli et al.~\cite[Section~3.2]{ameli2026} with edge costs $w$.
The algorithm constructs a minimum-cost strongly connected spanning subgraph through an ear decomposition, starting with a directed cycle and adding one ear at a time.
Use the following tables, with impossible entries set to infinity:
\begin{itemize}
\item For $u,v\in Z$ and $X\subseteq Z\setminus\{u,v\}$, $\mathsf P[u,v,X]$ stores the minimum cost of a $u$-to-$v$ path in $J$ whose internal vertex set is exactly $X$.
Compute these entries by Held--Karp dynamic programming~\cite{held1962}.
\item For $1\le q\le n-1$ and $X\subseteq Z$, $\mathsf D[q,X]$ stores the minimum cost of an ear decomposition on $X$ with $q$ ears, each ear after the first introducing at least one vertex.
\end{itemize}
Initialize $\mathsf D[1,X]=\min_{u\in X}\mathsf P[u,u,X\setminus\{u\}]$.
For $q\ge2$, append an ear with nonempty internal vertex set $Y\subseteq X$ by choosing its endpoints $u,v\in X\setminus Y$ and minimizing $\mathsf D[q-1,X\setminus Y]+\mathsf P[u,v,Y]$ over these choices.
For each fixed pair of endpoints, this is a $(\min,+)$ subset convolution, which can be computed in $\Ostar{2^n}$ time when entries are polynomially bounded \cite{bjorklund2007}.
Every transition appends an ear, and removing the last ear of a decomposition gives a transition that constructs it.

For unweighted \sca, apply the dynamic programming to the single-cost instance $\mathcal J=(J,d_H,k)$ from \Cref{subsec:results}.
Every decomposition considered uses at most $2n-2$ edges, and the costs $d_H$ count links on paths in $G+L$, so all finite table entries are polynomially bounded in the input size.
The dynamic programming therefore runs in $\Ostar{2^n}$ time, giving $\Ostar{4^k}$ time for unweighted \sca.
We accept if and only if $\min_{1\le q\le n-1}\mathsf D[q,Z]\le k$.

\unweightedtheorem*

\weightedtheorem*

\begin{proof}[Proof of \Cref{thm:weighted}]
For general link weights, keep all edges of the instance $J$ given by \Cref{thm:reduction}.
Add an index $0\le h\le k$ to the preceding tables: $\mathsf P[u,v,X,h]$ and $\mathsf D[q,X,h]$ store the minimum second cost of the corresponding path or ear decomposition with first cost exactly $h$.
In each transition, add the first-cost indices and second costs.
We accept if and only if $\min_{1\le q\le n-1,\,0\le h\le k}\mathsf D[q,Z,h]\le t$.

Each recurrence can be computed in $\Ostar{3^n}$ time by directly enumerating disjoint vertex sets.
Since $n\le2k$, this gives $\Ostar{9^k}$ time for \sca.
\end{proof}

\subsection{Polynomial kernels}
\label{sec:kernel}

We construct a smaller \sca\ instance that gives the same terminal graph $J$ under the reduction in \Cref{sec:terminal-reduction}.
By \Cref{thm:reduction}, this ensures that the two \sca\ instances are equivalent.
We first prove the unweighted kernel, then handle general weights by retaining more paths and applying weight reduction.

\paragraph*{Construction.}
Let $\mathcal I=(G,L,w,k,t)$ be an instance of \sca, and replace $k$ by $\min\{k,|L|\}$, since a solution uses at most $|L|$ links.
Let $S$ and $T$ be the sources and sinks of $G$, and let $Z=S\cup T$.
As in \Cref{subsec:results}, reject if $|S|>k$ or $|T|>k$, so henceforth $|Z|\le2k$.
For every distinct $u,v\in Z$ and $0\le h\le k$, compute a minimum-weight $u$-to-$v$ path $P_{uv,h}$ using at most $h$ links whenever such a path exists, by the layered shortest-path computation in \Cref{lem:bounded-link-distances}.
Its weight is $d_h(u,v)$.
Retain all these paths; for unweighted \sca, retain only one using the fewest links among the computed paths for each ordered pair of terminals.
Let $K\subseteq L$ be the union of the links on the retained paths, and let $W:=Z\cup\bigcup_{xy\in K}\{x,y\}$.
Construct a directed graph $G_W$ with vertex set $W$ and a base edge $xy$ for every distinct $x,y\in W$ such that $G$ has an $x$-to-$y$ path.
The resulting instance is $(G_W,K,w|_K,k,t)$, or $(G_W,K,k)$ in the unweighted case.

\unweightedkernel*

\begin{proof}[Proof of \Cref{thm:kernel-unweighted}]
We show that the reduction to \textsc{SCSS} applied to $(G,L,k)$ and $(G_W,K,k)$ give the same terminal graph $J$.
The equivalence then follows from \Cref{thm:reduction}.
By construction, they have the same sources and sinks.

For distinct $u,v\in Z$ and $0\le h\le k$, we show that a $u$-to-$v$ path with at most $h$ links exists in $G+L$ if and only if one exists in $G_W+K$.
For the forward direction, the retained $u$-to-$v$ path uses the minimum number of links, hence at most $h$, all in $K$.
Each maximal nonempty segment of this path in $G$ has both endpoints in $W$.
Replacing each such segment by the corresponding base edge of $G_W$ gives the required path in $G_W+K$.
Conversely, replace each base edge on a $u$-to-$v$ path in $G_W+K$ by a path in $G$.
Deleting closed subwalks gives a $u$-to-$v$ path in $G+L$ without increasing the number of links, since $K\subseteq L$.
Thus all values $d_h(u,v)$ for $0\le h\le k$ are preserved, so both inputs give the same terminal graph $J$.

There are at most $|Z|(|Z|-1)$ retained paths, each using at most $k$ links, so $|K|\le |Z|(|Z|-1)k\le4k^3$ and $|W|\le |Z|+2|K|\le2k+8k^3$.
\end{proof}

For the weighted kernel, we use Frank--Tardos weight reduction~\cite{frank1987} in the formulation of Etscheid et al.~\cite[Theorem~1]{etscheid2015}.

\begin{lemma}[Frank--Tardos weight reduction]
\label{lem:weight-reduction}
Given integers $r,N\ge1$ and a vector $\xi\in\mathbb Q^r$, one can compute in polynomial time a vector $\bar\xi\in\mathbb Z^r$ with $\|\bar\xi\|_\infty\le2^{4r^3}N^{r(r+2)}$ such that $\operatorname{sign}(\xi\cdot z)=\operatorname{sign}(\bar\xi\cdot z)$ for every $z\in\mathbb Z^r$ with $\|z\|_1\le N-1$.
\end{lemma}

\weightedkernel*

\begin{proof}[Proof of \Cref{thm:kernel}]
Use the construction above with a path $P_{uv,h}$ for every distinct $u,v\in Z$ and $0\le h\le k$ with $d_h(u,v)<\infty$.
We show that $(G,L,w,k,t)$ and $(G_W,K,w|_K,k,t)$ give the same terminal graph with both costs.
The terminal set remains $Z$, as in the unweighted case.
For distinct $u,v\in Z$ and $0\le h\le k$ with $d_h(u,v)<\infty$, compressing the base-edge segments of $P_{uv,h}$ gives a path in $G_W+K$ of weight at most $d_h(u,v)$ using at most $h$ links.
Conversely, expanding base edges and deleting closed subwalks transforms any path in $G_W+K$ into a path in $G+L$ without increasing its link count or weight.
Thus every value $d_h(u,v)$, including infinity, is preserved, so both instances give the same $J$ and are equivalent by \Cref{thm:reduction}.

There are at most $|Z|(|Z|-1)(k+1)=O(k^3)$ retained paths, each using at most $k$ links, so $|K|=O(k^4)$ and $|W|=O(k^4)$.
Relabeling $W$ consecutively bounds the graph encoding by $O(k^8\log k)$ bits.
It remains to bound the encoding length of the weights and budget.

Fix an ordering of $K$, let $r=|K|+1$, and let $\xi=((w(e))_{e\in K},t)\in\mathbb Z_{\ge0}^r$.
Apply \Cref{lem:weight-reduction} to $\xi$ with $N=k+2$, and write the resulting vector as $\bar\xi=((\bar w(e))_{e\in K},\bar t)$.
Applying sign preservation to the coordinate unit vectors shows that $\bar w$ and $\bar t$ are nonnegative.

For every $F\subseteq K$ with $|F|\le k$, let $\chi_F\in\{0,1\}^{|K|}$ be its incidence vector.
The vector $(\chi_F,-1)$ has $\ell_1$-norm at most $k+1\le N-1$, so $w(F)\le t$ if and only if $\bar w(F)\le\bar t$.
Thus $(G_W,K,\bar w,k,\bar t)$ is an equivalent instance of \sca.
By the norm bound in \Cref{lem:weight-reduction}, the vector $\bar\xi$ uses $O(r^4+r^3\log N)=O(k^{16})$ bits in total, which also bounds the complete kernel encoding.
\end{proof}


\section{Conclusion}
\label{sec:conclusion}

We obtained single-exponential algorithms and polynomial kernels for \sca\ by reducing it to \textsc{SCSS} with two edge costs.
In the undirected setting, similar reductions to \textsc{Steiner Tree} have also proved useful for approximation algorithms.
Basavaraju et al.~\cite{basavaraju2014} reduced edge-connectivity augmentation to \textsc{Steiner Tree} to obtain single-exponential algorithms.
Byrka, Grandoni, and Jabal Ameli~\cite{byrka2020} later exploited this reduction to obtain a $1.91$-approximation for the unweighted problem, the first ratio below $2$.
Subsequently, Traub and Zenklusen~\cite{traub2023} obtained a $(1.5+\varepsilon)$-approximation for the weighted problem, for every fixed $\varepsilon>0$, giving the first ratio below $2$ in this setting.
Their approach is closely related to the vertex-weighted \textsc{Steiner Tree} formulation.
The success of the \textsc{Steiner Tree} reduction suggests that our \textsc{SCSS} reduction may also be useful for approximation algorithms for \sca.

\section*{Acknowledgements}
Tomohiro Koana was supported in part by JST CREST Grant Number JPMJCR24Q2 and JST ERATO Grant Number JPMJER2301.

\section*{Declaration of generative AI use}
ChatGPT 6 Astra proposed a proof of the unweighted kernel result (\Cref{thm:kernel-unweighted}).
The authors extracted from this proof the reduction to \textsc{SCSS} with two edge costs and used it to derive the single-exponential algorithms and the polynomial kernel for general link weights.
ChatGPT 6 Astra was also used to draft the manuscript.
The authors verified and revised the manuscript and take full responsibility for it.

\begingroup
\small

\endgroup
\end{document}